\documentclass[letterpaper,10pt,conference]{ieeeconf}

\IEEEoverridecommandlockouts
\usepackage{cite}
\usepackage{amsmath,amssymb,amsfonts,mathtools}
\usepackage{graphicx}
\usepackage{textcomp}
\usepackage{comment}
\usepackage{xcolor}

\usepackage{amsthm}

\newtheorem{theorem}{Theorem}
\newtheorem{corollary}[theorem]{Corollary}

\newtheorem{lemma}[theorem]{Lemma}
\newtheorem{proposition}[theorem]{Proposition}
\newtheorem{remark}{Remark}
\newtheorem{problem}{Problem}

\newtheorem{example}{Example}

\usepackage[hidelinks]{hyperref}

\newcommand{\R}{\mathbb{R}}
\DeclareMathOperator*{\argmin}{arg\,min}

\def\BibTeX{{\rm B\kern-.05em{\sc i\kern-.025em b}\kern-.08em
    T\kern-.1667em\lower.7ex\hbox{E}\kern-.125emX}}

\begin{document}
\bstctlcite{BSTcontrol}

\title{\LARGE \bf Feasibility and Explicit Safety Filters for Control Barrier Functions  \\in Linear Systems}

\author{%
\makebox[\textwidth][c]{%
\mbox{Shima Sadat Mousavi, Max H. Cohen, Pol Mestres, and Aaron D. Ames}}%
\thanks{S. S. Mousavi, P. Mestres, and A. D. Ames are with the
Department of Mechanical and Civil Engineering, California Institute
of Technology, Pasadena, CA 91125, USA.
Emails: \texttt{smousavi,mestres,ames@caltech.edu}.}%
\thanks{M. H. Cohen is with the Department of Electrical and Computer
Engineering, North Carolina State University, Raleigh, NC 27695, USA.
Email: \texttt{mhcohen2@ncsu.edu}.}%
\thanks{This research was supported by the Boeing Strategic University Initiative.}%
}

\maketitle

\begin{abstract}
Safety filters based on control barrier functions (CBFs) and high-order
control barrier functions (HOCBFs) are commonly implemented through
quadratic programs (QPs), which may become infeasible under conflicting
safety and actuator constraints. For linear systems with affine safety
functions, we exploit repeated input-space directions to obtain an
explicit polyhedral feasibility domain, linear-program certificates for
regional feasibility and final-gain selection, and conditions identifying
gain-independent obstructions and the minimum common final gain. A
compatible QP metric yields an exact saturation-based filter and, together
with the regional certificate, persistent feasibility and safety. For a
single exterior ellipsoidal CBF with unrestricted inputs, we also derive
an exact, gain-independent linear matrix inequality (LMI) test.
Double-integrator and aircraft examples illustrate the results.
\end{abstract}

\section{Introduction}

Safety filters based on control barrier functions (CBFs) and high-order
control barrier functions (HOCBFs) are commonly implemented through
quadratic programs (QPs) that minimally modify a nominal control input to
enforce safety \cite{AmesTAC17,WX-CB:22,GurrietICCPS18}. With multiple
safety requirements and actuator limits, however, the constraints may
conflict, restricting QP feasibility to only part of the state space.
Characterizing this region offline, certifying feasibility over a
prescribed set, and identifying cases where the QP admits an explicit
solution are important for certification and real-time control.

Feasibility of multiple CBF constraints has been studied through
control-sharing conditions, barrier compositions, sampling-based tests,
and optimization-based certificates
\cite{GlotfelterLCSS17,XuAutomatica18,TamasLCSS23,
CortezLCSS22,CortezAutomatica22,TanCDC22,BreedenACC23,IsalyTAC24,
mousavi2025vertices}. Regional certificates include vertex-based linear
programming (LP) tests \cite{mousavi2025vertices} and controller-free
polynomial verification \cite{Dai2024Compatibility}, while scalar
interval feasibility and saturation also arise in robust filtering
\cite{Fikri2026RobustScalar}. Here, we exploit the case where several
barrier constraints share the same input-space direction. Grouping such
constraints into parallel families permits direct elimination of the
input variables, yielding the feasibility domain without generic
polyhedral projection or state-space vertex enumeration and enabling
regional certification and HOCBF final-gain selection. Maximal output
admissible sets and reference governors provide related constraint
guarantees for pre-stabilized linear systems
\cite{GilbertTan91,KolmanovskyAutomatica17}, but modify the reference
based on predicted closed-loop behavior; our focus is instead on the
instantaneous barrier inequalities and their compatibility throughout
an admissible state set.

Closed-form safety controllers have also been derived under particular
constraint structures. In \cite{lavretsky2025servo}, operational limits
for linear time-invariant (LTI) systems are enforced through componentwise
bounds on an output with the same dimension as the input, while
\cite{cohen2025compatibility} derives closed-form multi-CBF controllers
for nonlinear systems with a square input--output structure using a
decoupling matrix in the QP cost. More general analytical safety-filter
formulas are studied in \cite{cohen2023smooth}, and explicit CBF-QPs can
be obtained through state-space partitioning \cite{mestres2025explicit}.
The distinction here is that several distinct safety functions may share
one input direction while imposing different, potentially conflicting,
bounds. Such families arise, for example, in flight-envelope protection
and need not correspond to the two bounds of a single output. This
directional structure supports both regional gain selection and a
lower-dimensional explicit safety filter.

This paper primarily considers LTI systems with affine safety functions,
for which the HOCBF inequalities have fixed input-space directions and
affine state dependence. We additionally consider a complementary
quadratic-CBF setting in which the input direction is state-dependent
and may vanish. The main contributions are:
\begin{itemize}

\item We derive exact feasibility conditions for independent families of
parallel input-space normals. For affine HOCBFs, these yield an
explicit polyhedral feasibility domain and an exact LP for regional
certification and final-gain selection. On compact polytopic
HOCBF-admissible sets, we identify gain-independent obstructions and
the minimum feasible common final gain.

\item We reduce the safety-filter QP to the independent constraint
directions and characterize the cost-function weights that yield an exact
componentwise-saturation solution. Combined with the regional certificate,
this filter guarantees persistent feasibility and safety on the
HOCBF-admissible set.

\item For a single exterior ellipsoidal CBF with unrestricted inputs,
we derive an exact linear matrix inequality (LMI) test for feasibility
throughout the safe set, independent of the positive CBF gain, and a projected-matrix condition when the ellipsoid center is a controlled equilibrium.

\end{itemize}

\section{Problem Formulation}
\label{sec:problem}

\noindent\textbf{Notation.}
For \(v\in\mathbb R^p\), \(v_i\) denotes its \(i\)th entry.
Vector inequalities are understood componentwise. For \(a\leq b\),
possibly with infinite endpoints, let
\(
[a,b]:=\{z:a\leq z\leq b\}
\)
and
\(
\operatorname{sat}_{[a,b]}(z):=\min\{\max\{z,a\},b\},
\)
with componentwise extension to vectors. We use \(I_n\) for the
\(n\times n\) identity matrix.

Consider the LTI system
\begin{equation}
\dot{x}=Ax+Bu,
\label{eq:dyn}
\end{equation}
where \(x\in\mathbb R^n\) and \(u\in\mathcal U\subseteq\mathbb R^m\).
The admissible input set \(\mathcal U\) is either \(\mathbb R^m\) or a
nonempty polyhedron representing actuator constraints, including box
constraints. Let \(u_d(x)\) be an arbitrary nominal control input.

We first consider safety specified by the affine constraints
\begin{equation}
h_i(x)=a_i^\top x-b_i\geq0,\qquad i=1,\ldots,p,
\label{eq:affine_h}
\end{equation}
where \(a_i\in\mathbb R^n\) and \(b_i\in\mathbb R\). These 
define the safe set
\[
\mathcal C\coloneqq
\{x\in\mathbb R^n:h_i(x)\geq0,\ i=1,\ldots,p\}.
\]
Assume that each \(h_i\) has finite relative degree \(r_i\geq1\), meaning
that
\[
a_i^\top A^kB=0,\quad k=0,\ldots,r_i-2,
\qquad
a_i^\top A^{r_i-1}B\neq0.
\]
For \(r_i=1\), the first condition is vacuous.
We enforce \eqref{eq:affine_h} using HOCBF constraints \cite{WX-CB:22}
with linear class-\(\mathcal K\) gains
\(s\mapsto\alpha_{i,k}s\), where \(\alpha_{i,k}>0\). Define
\[
\psi_{i,0}\coloneqq h_i,\quad
\psi_{i,k+1}
\coloneqq
\dot{\psi}_{i,k}+\alpha_{i,k+1}\psi_{i,k},
\: k=0,\ldots,r_i-1.
\]
The derivative is taken along \eqref{eq:dyn}. By the relative-degree
assumption, \(\psi_{i,k}\) depends only on \(x\) for \(k<r_i\), whereas
\(\psi_{i,r_i}\) also depends on \(u\).
The HOCBF-admissible set is
\begin{equation}
\mathcal S\coloneqq
\{x:\psi_{i,k}(x)\geq0,\;
i=1,\ldots,p,\ k=0,\ldots,r_i-1\}.
\label{eq:S_set}
\end{equation}
Assume  \(\mathcal S\) is nonempty. The final HOCBF constraints are
\begin{equation}
\psi_{i,r_i}(x,u)\geq0,\qquad i=1,\ldots,p.
\label{eq:hocbf_condition}
\end{equation}
Under the standard HOCBF regularity conditions, if a feedback
$u=k(x)$ is locally Lipschitz on a neighborhood of $\mathcal S$,
satisfies $k(x)\in\mathcal U$, and satisfies
\eqref{eq:hocbf_condition} for every $x\in\mathcal S$, then
trajectories initialized in $\mathcal S$ remain in
$\mathcal S\subseteq\mathcal C$ throughout their interval of existence.

For each \(i\in\{1,\ldots,p\}\), define
\[
\Phi_i\coloneqq
\prod_{j=1}^{r_i}(A+\alpha_{i,j}I_n),
\qquad
\bar\alpha_i\coloneqq
\prod_{j=1}^{r_i}\alpha_{i,j}.
\]
By \cite[Lemma~2]{PM-SSM-ADA:26}, \eqref{eq:hocbf_condition} is
equivalent to
\begin{equation}
\ell_i^\top u+\beta_i(x)\geq0,
\label{eq:stacked}
\end{equation}
where
\[
\ell_i^\top\coloneqq a_i^\top A^{r_i-1}B,
\qquad
\beta_i(x)\coloneqq a_i^\top\Phi_i x-\bar\alpha_i b_i.
\]
Thus, \(\ell_i\neq0\) is constant and \(\beta_i(x)\) is affine in \(x\).
We call \(\ell_i\) the \emph{input-space constraint normal}, distinct
from the state-space normal \(a_i\) in \eqref{eq:affine_h}.
Linear HOCBF gains keep the offsets affine; nonlinear
class-\(\mathcal K\) choices generally do not.

Stacking the HOCBF constraints \eqref{eq:stacked} gives
\begin{equation}
Mu\leq d(x),
\label{eq:affine_constraints}
\end{equation}
where the \(i\)th row of \(M\) is \(-\ell_i^\top\) and
\(d_i(x)=\beta_i(x)\). Thus, \(M\) is constant and \(d(x)\) is affine.
The resulting safety filter is
\begin{equation}
\begin{aligned}
u^\star(x)=\argmin_{u\in\mathcal U}\quad&
\frac12\big(u-u_d(x)\big)^\top G\big(u-u_d(x)\big)\\
\mathrm{s.t.}\quad& Mu\leq d(x),
\end{aligned}
\label{eq:safety_qp}
\end{equation}
where \(G\succ0\). Define the feasible input set and feasibility domain as
\begin{equation}
\begin{aligned}
\mathcal F(x)&\coloneqq\{u\in\mathcal U:Mu\leq d(x)\},\\
\mathcal X_{\mathrm{feas}}
&\coloneqq\{x\in\mathbb R^n:\mathcal F(x)\neq\emptyset\}.
\end{aligned}
\label{eq:feas_sets}
\end{equation}
The filter has a unique solution for every
\(x\in\mathcal X_{\mathrm{feas}}\), since the objective in
\eqref{eq:safety_qp} is coercive and strictly convex on a nonempty
closed feasible set. Feasibility depends only on the HOCBF and
actuator constraints, not on \(u_d\) or \(G\). In particular,
\(\mathcal S\subseteq\mathcal X_{\mathrm{feas}}\) ensures that the
filter is feasible throughout the HOCBF-admissible set.

By an \emph{explicit} characterization of
$\mathcal X_{\mathrm{feas}}$, we mean a finite collection of
inequalities in $x$ that can be evaluated without solving a
state-dependent feasibility problem. We seek to construct these
inequalities directly from the constraint structure, without
state-space gridding or polyhedral projection.

For each constraint, we call \(\alpha_{i,r_i}\) the final HOCBF gain.
With the lower-order gains fixed, changing the final gains leaves
\(\mathcal S\) unchanged. We study the following problem.

\begin{problem}
\label{problem}
For the safety filter \eqref{eq:safety_qp}, identify classes of
$M$ and $\mathcal U$ for which the feasibility domain
$\mathcal X_{\mathrm{feas}}$ admits an explicit characterization.
For these classes, certify
$\mathcal S\subseteq\mathcal X_{\mathrm{feas}}$ and, with the
lower-order HOCBF gains fixed, determine whether positive final
gains ensuring this inclusion exist and select them when they do.
We also seek conditions for a closed-form safety filter.
\end{problem}

As a complementary setting, we also consider a single exterior quadratic
CBF
\[
h_q(x)=(x-x_c)^\top P(x-x_c)-\rho,\qquad P\succ0,\ \rho>0,
\]
with unrestricted inputs. Unlike the affine case, its input-space normal
depends on the state and may vanish. We seek an exact condition for
feasibility throughout its safe set.

If all safety constraints have relative degree one, then
\(\mathcal S=\mathcal C\). For higher-relative-degree constraints,
the HOCBF safety guarantee applies to trajectories initialized in
\(\mathcal S\). The following example motivates the directional
structure.

\begin{example}
\label{ex:double_integrator}
Consider the double integrator
\[
\dot x_1=x_2,\qquad \dot x_2=u,
\]
with \(x=[\,x_1\ x_2\,]^\top\) and
\(u\in\mathcal U\subseteq\mathbb R\). Let
\[
h_i(x)=a_{i,1}x_1+a_{i,2}x_2-b_i,\qquad i=1,\ldots,p.
\]
With linear HOCBF gains, the HOCBF constraint associated with each
\(h_i\) takes the form \eqref{eq:stacked}. If \(a_{i,2}\neq0\), then
\(h_i\) has relative degree one, with
\[
\ell_i=a_{i,2},\qquad
\beta_i(x)=a_{i,1}x_2+
\alpha_{i,1}(a_{i,1}x_1+a_{i,2}x_2-b_i).
\]
If \(a_{i,2}=0\) and \(a_{i,1}\neq0\), then \(h_i\) has relative degree
two, with \(\ell_i=a_{i,1}\) and
\[
\beta_i(x)=
(\alpha_{i,1}+\alpha_{i,2})a_{i,1}x_2
+\alpha_{i,1}\alpha_{i,2}(a_{i,1}x_1-b_i).
\]
Here, arbitrarily many distinct affine safety functions share one
input direction. Their offsets determine whether the resulting
lower and upper bounds are compatible.
\end{example}

\section{Affine Feasibility and Gain Selection}
\label{sec:feasibility}
In this section, we first recall a general feasibility characterization
for polyhedral actuator constraints. We then exploit repeated HOCBF
constraint directions to obtain an explicit feasibility domain, and use
this characterization to certify regional feasibility and select the
final HOCBF gains.

Suppose the admissible input set is polyhedral:
\[
\mathcal U=\{u\in\mathbb R^m:Qu\leq b\},
\qquad Q\in\mathbb R^{r_u\times m}.
\]
Then, the safety filter is feasible at \(x\in\R^n\) if and only if there exists
\(u\in\R^m\) satisfying:
\begin{equation}
Mu\leq d(x),\qquad Qu\leq b.
\label{eq:stacked_feas}
\end{equation}

\begin{lemma}
\label{lem:farkas}
The constraints \eqref{eq:stacked_feas} are feasible if and only if:
\[
\lambda^\top
\begin{bmatrix}
d(x)\\ b
\end{bmatrix}
\geq0
\]
for every \(\lambda\in\mathbb R_{\geq0}^{p+r_u}\) satisfying:
\[
\lambda^\top
\begin{bmatrix}
M\\ Q
\end{bmatrix}
=0.
\]
\end{lemma}

\begin{proof}
Apply Farkas' lemma \cite[Sec.~5.8]{Boyd} to \eqref{eq:stacked_feas}.
\end{proof}

%Since \(d(x)\) is affine, the set of \((x,u)\) satisfying \eqref{eq:stacked_feas} is polyhedral, and so is its projection \(\mathcal X_{\mathrm{feas}}\). In general,  obtaining an explicit state-space representation requires eliminating the input variables. We next show that repeated constraint directions allow this representation to be obtained directly.

Since \(d(x)\) is affine, \(\mathcal X_{\mathrm{feas}}\) is polyhedral.
Lemma~\ref{lem:farkas} can, in principle, yield a finite state-space
description by characterizing the nonnegative \(\lambda\) satisfying
\(\lambda^\top[M^\top\;Q^\top]^\top=0\), but this generally requires an
additional offline computation. We next show that repeated constraint
directions yield these inequalities directly.

\subsection{Repeated Constraint Directions}

We now consider the case where the \(p\) HOCBF constraints
\eqref{eq:stacked} act through only \(q\leq m\) independent constraint
directions. Let \(T_1,\ldots,T_q\) partition
\(\{1,\ldots,p\}\) such that:
\begin{equation}
\ell_i=c_i v_j,\qquad \forall i\in T_j,
\label{eq:directional_structure}
\end{equation}
where \(c_i\neq0\), \(v_j\neq0\), and
\(v_1,\ldots,v_q\) are linearly independent. We call \(T_j\) a
\emph{parallel family}. Thus, all HOCBF constraints in \(T_j\) depend
on the input through the same scalar quantity \(v_j^\top u\), where
\(v_j\) is their common \emph{constraint direction}.

For each parallel family, define:
\[
T_j^+:=\{i\in T_j:c_i>0\},
\qquad
T_j^-:=\{i\in T_j:c_i<0\},
\]
and:
\begin{equation}
\nu_i(x):=-\frac{\beta_i(x)}{c_i}.
\label{eq:nu}
\end{equation}
The tightest lower and upper bounds along \(v_j\) are:
\begin{equation}
\underline s_j(x):=
\max_{i\in T_j^+}\nu_i(x),
\qquad
\overline s_j(x):=
\min_{i\in T_j^-}\nu_i(x),
\label{eq:directional_bounds}
\end{equation}
where \(\max\emptyset=-\infty\) and
\(\min\emptyset=+\infty\).

Collect the constraint directions in:
\begin{equation}
S=[\,v_1\ \cdots\ v_q\,]^\top
\in\mathbb R^{q\times m},
\label{eq:S_definition}
\end{equation}
and their bounds in:
\begin{equation}
\underline s(x)=
[\,\underline s_1(x)\ \cdots\ \underline s_q(x)\,]^\top,
\quad
\overline s(x)=
[\,\overline s_1(x)\ \cdots\ \overline s_q(x)\,]^\top.
\label{eq:s_vectors}
\end{equation}
Since the constraint directions are linearly independent, \(S\) has full
row rank.
To account for actuator constraints, define:
\begin{equation}
\mathcal Z_{\mathcal U}:=
\{Su:u\in\mathcal U\}\subseteq\mathbb R^q.
\label{eq:ZU}
\end{equation}
Thus, \(\mathcal Z_{\mathcal U}\) contains all values of
\(Su=[\,v_1^\top u\ \cdots\ v_q^\top u\,]^\top\) that can be generated
while satisfying the actuator constraints. In general, the components of
\(Su\) cannot be chosen independently. We focus in particular on the case where:
\begin{equation}
\mathcal Z_{\mathcal U}
=
[\underline z_u,\overline z_u]
=
\{z\in\mathbb R^q:\underline z_u\leq z\leq\overline z_u\},
\label{eq:ZU_box}
\end{equation}
where \(\underline z_u\) and \(\overline z_u\) are \(q\)-dimensional
vectors of lower and upper bounds, possibly with infinite entries. In this
case, each \(v_j^\top u\) can take any value within its interval
independently of the others. For example, if \(\mathcal U\) is an input
box and \(v_j=e_j\), then \(v_j^\top u=u_j\), so
\(\mathcal Z_{\mathcal U}\) is also a box.

The definitions above replace each parallel family by lower and upper
requirements on one scalar input direction. With unbounded inputs, the
independence of the constraint directions is enough to separate these
requirements. With actuator constraints, feasibility also depends on the
values of \(Su\) allowed by \(\mathcal U\). The next result makes this
precise.

\begin{theorem}
\label{thm:directional_feasibility}
Suppose the \(p\) HOCBF constraints form \(q\leq m\) parallel families
with linearly independent constraint directions as in
\eqref{eq:directional_structure}. Then the \(p\) HOCBF constraints in \eqref{eq:stacked} are jointly
equivalent to:
\begin{equation}
\underline s(x)\leq Su\leq\overline s(x).
\label{eq:directional_box}
\end{equation}
Moreover:
\begin{enumerate}
\item[(i)] If \(\mathcal U=\mathbb R^m\), the QP \eqref{eq:safety_qp} is feasible at \(x\) if and only if:
\begin{equation}
\underline s(x)\leq\overline s(x).
\label{eq:unbounded_directional_feas}
\end{equation}

\item[(ii)] For a general admissible input set \(\mathcal U\), the safety
filter is feasible at \(x\) if and only if:
\begin{equation}
[\underline s(x),\overline s(x)]
\cap\mathcal Z_{\mathcal U}\neq\emptyset .
\label{eq:general_directional_feas}
\end{equation}

\item[(iii)] If \(\mathcal Z_{\mathcal U}\) has the box form
\eqref{eq:ZU_box}, the QP \eqref{eq:safety_qp} is feasible at \(x\)
if and only if:
\begin{equation}
\max\{\underline s(x),\underline z_u\}
\leq
\min\{\overline s(x),\overline z_u\},
\label{eq:box_directional_feas}
\end{equation}
where the maximum and minimum are componentwise.
\end{enumerate}
\end{theorem}

\begin{proof}
For \(i\in T_j\), the HOCBF constraint \eqref{eq:stacked} becomes:
\[
c_i v_j^\top u+\beta_i(x)\geq0.
\]
If \(c_i>0\), the above implies that
\(v_j^\top u\geq\nu_i(x)\); if \(c_i<0\), it implies
\(v_j^\top u\leq\nu_i(x)\). Taking the largest lower bound and smallest
upper bound in each parallel family then yields
\eqref{eq:directional_box}.

For (i), full row rank of $S$ implies that every $z\in\R^q$
can be written as $Su$. Thus a feasible input exists exactly when
the interval box is nonempty, i.e.,
$\underline s(x)\leq\overline s(x)$. This also covers missing
one-sided bounds.

For (ii), necessity follows from the observation that any feasible \(u\in\mathcal U\)
satisfies
\(Su\in[\underline s(x),\overline s(x)]\) and
\(Su\in\mathcal Z_{\mathcal U}\). Conversely, if
\(z\in[\underline s(x),\overline s(x)]\cap\mathcal Z_{\mathcal U}\),
then \(z=Su\) for some \(u\in\mathcal U\), and this input satisfies
\eqref{eq:directional_box}.

For (iii), if $\mathcal{Z}_{\mathcal{U}}$ satisfies \eqref{eq:ZU_box}, then \eqref{eq:general_directional_feas} reduces to the intersection of two boxes
\([\underline s(x),\overline s(x)]\) and
\([\underline z_u,\overline z_u]\). Their intersection is nonempty if
and only if \eqref{eq:box_directional_feas} holds componentwise.
\end{proof}

Theorem~\ref{thm:directional_feasibility} reduces the \(p\) HOCBF
constraints to \(q\) interval constraints. With unbounded inputs, or when
\(\mathcal Z_{\mathcal U}\) is a box, feasibility can be checked
independently along each direction. The result also recovers the fully
parallel and fully independent cases: if all input-space normals are
parallel, \(q=1\) and feasibility reduces to a single interval test; if
they are linearly independent, each family is a singleton and, with
unbounded inputs, the filter is feasible for every state. For \(q=1\),
\(\mathcal Z_{\mathcal U}\) is an interval for any polyhedral
\(\mathcal U\). If, in addition:
\[
\mathcal U=
\{u:u_k^{\min}\leq u_k\leq u_k^{\max},\ k=1,\ldots,m\},
\]
then \(\mathcal Z_{\mathcal U}=[\underline z_u,\overline z_u]\), where:
\[
\begin{aligned}
\underline z_u
=\sum_{k=1}^m
\min\{v_{1,k} u_k^{\min},\,v_{1,k} u_k^{\max}\},\\
\overline z_u
=\sum_{k=1}^m
\max\{v_{1,k} u_k^{\min},\,v_{1,k} u_k^{\max}\}.
\end{aligned}
\]
which captures the situation seen in Example \ref{ex:double_integrator}.
% Ex.~\ref{ex:double_integrator} is of this type.

We now specialize Theorem~\ref{thm:directional_feasibility} to 
LTI systems with affine constraints, 
% setting
and obtain the feasibility domain directly in the
state space. For each parallel family, define:
\begin{equation}
\mathcal X_j^{\rm u}:=
\bigcap_{i\in T_j^+}
\bigcap_{k\in T_j^-}
\{x:\nu_i(x)\leq\nu_k(x)\}.
\label{eq:Xj_unbounded}
\end{equation}
If \(\mathcal Z_{\mathcal U}\) has the box form
\eqref{eq:ZU_box}, also define
\begin{equation}
\begin{aligned}
\mathcal X_j^{\rm b}:={}&\mathcal X_j^{\rm u}
\cap\!\!\bigcap_{i\in T_j^+}
\{x:\nu_i(x)\leq\overline z_{u,j}\}\\
&{}\cap\!\!\bigcap_{k\in T_j^-}
\{x:\underline z_{u,j}\leq\nu_k(x)\}.
\end{aligned}
\label{eq:Xj_bounded}
\end{equation}
An intersection over an empty index set is understood as
\(\mathbb R^n\).

\begin{corollary}
\label{cor:explicit_feas_domain}
Consider the LTI system \eqref{eq:dyn} with affine safety constraints
\eqref{eq:affine_h}, and suppose the resulting HOCBF constraints satisfy
\eqref{eq:directional_structure}. 
% Then:
% \begin{enumerate}
% \item[(i)] 
If \(\mathcal U=\mathbb R^m\), then:
\begin{equation}
\mathcal X_{\mathrm{feas}}
=
\bigcap_{j=1}^q\mathcal X_j^{\rm u},
\label{eq:Xfeas_unbounded}
\end{equation}
and
% \item[(ii)] 
if \(\mathcal Z_{\mathcal U}\) has the box form
\eqref{eq:ZU_box}, then:
\begin{equation}
\mathcal X_{\mathrm{feas}}
=
\bigcap_{j=1}^q\mathcal X_j^{\rm b}.
\label{eq:Xfeas_bounded}
\end{equation}
% \end{enumerate}
In both cases, \(\mathcal X_{\mathrm{feas}}\) is a
polyhedron in the state space.
\end{corollary}

\begin{proof}
Each lower endpoint must not exceed each upper endpoint in its family.
For unrestricted inputs these comparisons are precisely
\eqref{eq:Xj_unbounded}. With the attainable box, comparisons with
the actuator endpoints add \eqref{eq:Xj_bounded}. Each comparison is
affine in $x$, proving the polyhedral description.
\end{proof}

Corollary~\ref{cor:explicit_feas_domain} provides a direct offline
construction of the feasibility domain by grouping parallel input-space
constraint normals, forming the affine functions \(\nu_i\), and collecting
the halfspace constraints in \eqref{eq:Xj_unbounded} or
\eqref{eq:Xj_bounded}. For unbounded inputs, this yields
\(\sum_{j=1}^q |T_j^+||T_j^-|\) halfspaces; if
\(\mathcal Z_{\mathcal U}\) has the box form \eqref{eq:ZU_box}, at most
\(p\) additional halfspaces account for the actuator constraints. Unlike
the general condition in Lemma~\ref{lem:farkas}, no polyhedral projection
is required.

\begin{remark}
\label{rem:nonlinear_directional}
Theorem~\ref{thm:directional_feasibility} does not require affine
offsets \(\beta_i(x)\), and thus extends to nonlinear control-affine
systems with fixed constraint directions over the region of interest.
The feasibility domain is then generally nonpolyhedral. If the directions
depend on the state, the reduction remains valid pointwise wherever the
directional structure and rank conditions hold.
\end{remark}

\subsection{Regional Certification and Final-Gain Selection}
\label{subsec:gain_selection}

Corollary~\ref{cor:explicit_feas_domain} gives the feasibility domain
for fixed HOCBF gains. We now turn to gain selection: with the
lower-order gains fixed, we seek final HOCBF gains such that a
prescribed region, in particular the HOCBF-admissible set $\mathcal S$,
is contained in the feasibility domain. We also characterize when a
single common final gain can achieve this property.

Fix the lower-order gains
$\alpha_{i,1},\ldots,\alpha_{i,r_i-1}$ and denote the final gains by
$\gamma_i:=\alpha_{i,r_i}$. Since
\[
\psi_{i,r_i}
=
\dot\psi_{i,r_i-1}
+\gamma_i\psi_{i,r_i-1},
\]
the final HOCBF constraint can be written as
\[
\ell_i^\top u+\eta_i(x)
+\gamma_i\psi_{i,r_i-1}(x)\geq0,
\]
where
$\eta_i(x):=\nabla\psi_{i,r_i-1}(x)^\top Ax$ is affine in $x$.
With the lower-order gains fixed, the set $\mathcal S$ is independent
of
\(
\gamma=[\,\gamma_1\ \cdots\ \gamma_p\,]^\top
\).

Under the conditions of
Corollary~\ref{cor:explicit_feas_domain}, the finite inequalities
defining the feasibility domain can therefore be written as
\begin{equation}
\mathcal X_{\mathrm{feas}}(\gamma)
=
\{x:F(\gamma)x\leq g(\gamma)\},
\label{eq:Xfeas_gain_form}
\end{equation}
where $F(\gamma)$ and $g(\gamma)$ depend affinely on $\gamma$.
Let
\[
\Omega=\{x:Ex\leq f\}
\]
be a fixed nonempty polyhedron. The next result gives an exact
condition for feasibility throughout $\Omega$.

\begin{theorem}[Regional feasibility]
\label{thm:regional_gain}
Under the conditions of
Corollary~\ref{cor:explicit_feas_domain}, for a given
$\gamma\in\mathbb R_{>0}^p$,
\[
\Omega\subseteq\mathcal X_{\mathrm{feas}}(\gamma)
\]
if and only if there exists an entrywise nonnegative matrix $\Lambda$
such that
\begin{equation}
\Lambda E=F(\gamma),\qquad
\Lambda f\leq g(\gamma).
\label{eq:regional_gain_certificate}
\end{equation}
\end{theorem}

\begin{proof}
If \eqref{eq:regional_gain_certificate} holds, then for every
$x\in\Omega$,
$F(\gamma)x=\Lambda Ex\leq\Lambda f\leq g(\gamma)$.
Conversely, if every row of
$F(\gamma)x\leq g(\gamma)$ is valid on the nonempty polyhedron
$\Omega$, LP duality \cite[Sec.~5.8]{Boyd} gives a nonnegative
multiplier of $Ex\leq f$ for each row. Stacking these multipliers
gives \eqref{eq:regional_gain_certificate}.
\end{proof}

Since $F$ and $g$ are affine in $\gamma$, the same certificate also
selects the final gains. For prescribed
$w,\underline\gamma\in\mathbb R_{>0}^p$, consider
\begin{equation}
\begin{aligned}
\min_{\gamma,\Lambda}\quad & w^\top\gamma\\
\mathrm{s.t.}\quad
&\Lambda E=F(\gamma),\qquad
\Lambda f\leq g(\gamma),\\
&\Lambda\geq0,\qquad
\gamma\geq\underline\gamma,
\end{aligned}
\label{eq:gain_selection_LP}
\end{equation}
which is an LP. If feasible, this LP yields a minimum-weight final-gain
vector guaranteeing feasibility throughout $\Omega$. In particular,
taking $\Omega=\mathcal S$ certifies
$\mathcal S\subseteq\mathcal X_{\mathrm{feas}}$.

Since
$\psi_{i,r_i-1}(x)\geq0$ on $\mathcal S$, increasing any final gain
can only relax the corresponding HOCBF constraint there. Hence, if one
gain vector guarantees feasibility throughout $\mathcal S$, every
componentwise larger gain vector also does.

We next consider the common-gain case, in which the same final gain is
used for all HOCBF constraints:
\[
\gamma_i=\gamma>0,\qquad i=1,\ldots,p.
\]
We determine whether a common gain can ensure feasibility throughout $\mathcal S$ and its minimum required value.

Suppose $\mathcal S$ is a nonempty compact polytope, and let $V$
denote its set of vertices.
Under a common final gain, each
inequality defining \eqref{eq:Xfeas_gain_form} can be written as
\begin{equation}
a_k(x)+\gamma s_k(x)\geq0,
\label{eq:common_gain_condition}
\end{equation}
where $a_k$ and $s_k$ are affine in $x$ and
$s_k(x)\geq0$ for all $x\in\mathcal S$. Define
\begin{equation}
\gamma_\star:=
\max\left\{
0,\;
\max_{\substack{k,\ v\in V\\s_k(v)>0}}
\frac{-a_k(v)}{s_k(v)}
\right\},
\label{eq:minimum_common_gain}
\end{equation}
where an empty inner maximum adds no lower bound.

The next result determines whether a sufficiently large common final
gain can guarantee feasibility throughout $\mathcal S$, and, when it
can, gives the minimum required gain.

\begin{proposition}[Existence and minimum common final gain]
\label{prop:gain_obstructions}
A finite positive common final gain satisfying
\[
\mathcal S\subseteq\mathcal X_{\mathrm{feas}}(\gamma)
\]
exists if and only if
\[
s_k(v)=0
\quad\Longrightarrow\quad
a_k(v)\geq0,
\qquad
\forall\,k,\ \forall\,v\in V.
\]
When this condition holds, the feasible positive common gains are
exactly
\[
\{\gamma>0:\gamma\geq\gamma_\star\}.
\]
\end{proposition}

\begin{proof}
If $s_k(v)=0$, condition \eqref{eq:common_gain_condition} reduces to
$a_k(v)\geq0$ and is therefore independent of $\gamma$. If
$s_k(v)>0$, it is equivalent to
$\gamma\geq-a_k(v)/s_k(v)$. Taking the largest of these bounds gives
\eqref{eq:minimum_common_gain}. Since
\eqref{eq:common_gain_condition} is affine in $x$, it holds throughout
the polytope $\mathcal S$ if and only if it holds at all its vertices.
\end{proof}

\section{Explicit Filters and Persistent Feasibility}
\label{sec:closed_form}

We now use the repeated-direction structure of
Theorem~\ref{thm:directional_feasibility} to obtain an explicit
solution of the safety-filter QP. We then combine this solution with
the regional feasibility certificate of
Theorem~\ref{thm:regional_gain} to establish persistent feasibility
and safety.

The \(p\) HOCBF constraints form \(q\leq m\) parallel families with
linearly independent directions \(v_1,\ldots,v_q\), collected in \(S\).
By Theorem~\ref{thm:directional_feasibility}, they are equivalent to
\eqref{eq:directional_box}.

For feasibility, it suffices to characterize the directional inputs
attainable under \(\mathcal U\). To obtain the QP optimizer explicitly,
we further require the actuator constraints to align with the same
directions:
\begin{equation}
\mathcal U=
\{u:\underline z_u\leq Su\leq\overline z_u\},
\label{eq:directional_input_set}
\end{equation}
where the bounds may have infinite components. This includes
\(\mathcal U=\mathbb R^m\) and, when \(S=I_m\), ordinary box
constraints.

Combining \eqref{eq:directional_input_set} with
\eqref{eq:directional_box}, define
\begin{equation}
\underline z(x)=
\max\{\underline s(x),\underline z_u\},\qquad
\overline z(x)=
\min\{\overline s(x),\overline z_u\}.
\label{eq:filter_bounds}
\end{equation}
The combined HOCBF and actuator constraints are therefore
\begin{equation}
\underline z(x)\leq Su\leq\overline z(x).
\label{eq:filter_directional_constraints}
\end{equation}

Define
\begin{equation}
z_d(x):=Su_d(x),\qquad
H:=SG^{-1}S^\top .
\label{eq:reduced_metric}
\end{equation}
Since \(S\) has full row rank and \(G\succ0\), we have \(H\succ0\).
For every feasible state, consider the reduced QP
\begin{equation}
\begin{aligned}
z^\star(x)=\argmin_z\quad&
\frac12(z-z_d(x))^\top H^{-1}(z-z_d(x))\\
\mathrm{s.t.}\quad&
\underline z(x)\leq z\leq\overline z(x),
\end{aligned}
\label{eq:reduced_qp}
\end{equation}
and define the componentwise saturation
\begin{equation}
\widehat z(x):=
\operatorname{sat}_{[\underline z(x),\,\overline z(x)]}
\bigl(z_d(x)\bigr).
\label{eq:z_sat}
\end{equation}

The next result reduces the \(m\)-dimensional safety-filter QP exactly
to \eqref{eq:reduced_qp} and characterizes when componentwise saturation
solves it for every nonempty interval box and nominal input.

\begin{theorem}[Explicit directional safety filter]
\label{thm:directional_filter}
Let the HOCBF constraints satisfy the conditions of
Theorem~\ref{thm:directional_feasibility}, and suppose
\eqref{eq:directional_input_set} holds. Then, for every
\(x\in\mathcal X_{\mathrm{feas}}\):
\begin{enumerate}
\item[(i)] The unique optimizer of \eqref{eq:safety_qp} is
\begin{equation}
u^\star(x)=u_d(x)
+G^{-1}S^\top H^{-1}
\bigl(z^\star(x)-z_d(x)\bigr),
\label{eq:general_directional_filter}
\end{equation}
where \(z^\star(x)\) solves \eqref{eq:reduced_qp}.

\item[(ii)] The solution of \eqref{eq:reduced_qp} satisfies
\(z^\star(x)=\widehat z(x)\) for every nonempty interval box and every
nominal input if and only if \(H\) is diagonal, equivalently,
\begin{equation}
v_i^\top G^{-1}v_j=0,\qquad i\neq j.
\label{eq:metric_orthogonality}
\end{equation}
In this case,
\begin{equation}
u^\star(x)=u_d(x)
+G^{-1}S^\top H^{-1}
\bigl(\widehat z(x)-z_d(x)\bigr).
\label{eq:explicit_directional_filter}
\end{equation}
\end{enumerate}
\end{theorem}

\begin{proof}
For (i), fix \(z=Su\). The KKT conditions for minimizing the quadratic
cost subject to \(Su=z\) give
\(u=u_d+G^{-1}S^\top H^{-1}(z-z_d)\). Substitution into the objective
yields \eqref{eq:reduced_qp}, and hence
\eqref{eq:general_directional_filter}.

For (ii), if \(H\) is diagonal, the reduced cost separates across the
\(q\) coordinates, so \eqref{eq:z_sat} is optimal. Conversely, let
\(W=H^{-1}\), set \(z_d=0\), and choose a box with
\(1\leq z_i\leq2\) and \(-1\leq z_k\leq1\) for \(k\neq i\).
If saturation is always optimal, then \(z^\star=e_i\). Since the
coordinates \(k\neq i\) are interior, first-order optimality gives
\(W_{ki}=0\). Repeating for all \(i\) shows that \(W\), and hence \(H\),
is diagonal.
\end{proof}

Condition \eqref{eq:metric_orthogonality} is automatic for \(q=1\);
otherwise, the directions must be pairwise \(G^{-1}\)-orthogonal.

\begin{remark}
\label{rem:G_construction}
Since \(S\) has full row rank, a positive-definite cost matrix
satisfying \eqref{eq:metric_orthogonality} can always be constructed.
A convenient choice is \(SG^{-1}S^\top=I_q\). If \(S\) is square and
invertible, one may take \(G=S^\top S\). Otherwise, for any \(\tau>0\),
\begin{equation}
G^{-1}=
S^\top(SS^\top)^{-2}S
+\tau\bigl(I_m-S^\top(SS^\top)^{-1}S\bigr),
\label{eq:G_construction}
\end{equation}
which is positive definite and satisfies \(SG^{-1}S^\top=I_q\).
\end{remark}

When \eqref{eq:metric_orthogonality} holds, online evaluation requires
only the directional bounds, \(q\) scalar saturations, and fixed matrix
operations. In contrast to explicit QP constructions based on
state-space partitioning \cite{mestres2025explicit}, no active-set
search or region-wise controller representation is required.

\begin{remark}
\label{rem:nonlinear_filter}
Theorem~\ref{thm:directional_filter} also applies to nonlinear
control-affine systems whose CBF or HOCBF constraints preserve the
required directional structure and whose actuator constraints satisfy
\eqref{eq:directional_input_set}. For state-dependent directions, the
reduction holds pointwise wherever \(S(x)\) has full row rank and
\eqref{eq:directional_input_set} holds with \(S=S(x)\). The saturation
formula further requires \(S(x)G^{-1}S(x)^\top\) to be diagonal.
\end{remark}

The regional certificate also ensures persistent feasibility.

\begin{corollary}[Persistent feasibility]
\label{cor:persistent_feasibility}
In the LTI-affine setting, suppose the hypotheses of
Theorem~\ref{thm:directional_filter} and
\eqref{eq:metric_orthogonality} hold, and $u_d$ is globally
Lipschitz. If \eqref{eq:regional_gain_certificate} holds with
$\Omega=\mathcal S$, then every $x(0)\in\mathcal S$ generates a
unique closed-loop solution satisfying
\[
x(t)\in\mathcal S\subseteq\mathcal C,\qquad
u^\star(x(t))\in\mathcal U,\qquad t\geq0.
\]
The safety filter therefore remains feasible for all time.
\end{corollary}

\begin{proof}
The certificate gives
$\mathcal S\subseteq\mathcal X_{\mathrm{feas}}$.
The max--min affine bounds make the explicit filter globally Lipschitz,
yielding a unique forward-complete solution.

Since $\mathcal S$ is polyhedral, Nagumo's condition reduces to
nonnegative derivatives of active affine constraints, with no constraint
qualification. The HOCBF recursion and final HOCBF constraint ensure
this, so $\mathcal S$ is forward invariant and the filter remains
feasible for all time.
\end{proof}

\section{Feasibility of a Single Quadratic CBF}
\label{sec:quadratic_cbf}

We now consider the complementary quadratic-CBF setting introduced in
Section~\ref{sec:problem}. Unlike the affine case, its input-space normal
depends on the state and may vanish. We show that feasibility over the
entire safe set is characterized exactly by an LMI.
To this end,
consider
\begin{equation}
h_q(x)=(x-x_c)^\top P(x-x_c)-\rho,
\qquad P=P^\top\succ0,\quad \rho>0,
\label{eq:quadratic_cbf}
\end{equation}
with safe set
\[
\mathcal C_q:=\{x:h_q(x)\geq0\}.
\]
For this section, let $\mathcal U=\mathbb R^m$. With a linear CBF gain
$\alpha>0$, the CBF constraint is
\begin{equation}
2(x-x_c)^\top P(Ax+Bu)+\alpha h_q(x)\geq0.
\label{eq:quadratic_cbf_constraint}
\end{equation}

Since the input coefficient in
\eqref{eq:quadratic_cbf_constraint} is
$2B^\top P(x-x_c)$, infeasibility can occur only at states satisfying
\begin{equation}
B^\top P(x-x_c)=0.
\label{eq:quadratic_zero_authority}
\end{equation}
Let $k:=\dim\ker(B^\top P)$. If $k=0$, the only state satisfying
\eqref{eq:quadratic_zero_authority} is
$x_c\notin\mathcal C_q$, and feasibility throughout $\mathcal C_q$ is
automatic. For $k\geq1$, let $N\in\mathbb R^{n\times k}$ have full
column rank with
$\operatorname{im}N=\ker(B^\top P)$, and define
\begin{equation}
\begin{aligned}
R_N&:=N^\top PN,\qquad
D_N:=N^\top(PA+A^\top P)N,\\
c_N&:=N^\top PAx_c .
\end{aligned}
\label{eq:quadratic_reduced_matrices}
\end{equation}
Since $P\succ0$, we have $R_N\succ0$.

The next result gives an exact condition for feasibility of the
quadratic CBF throughout $\mathcal C_q$.

\begin{proposition}[Quadratic CBF feasibility]
\label{prop:quadratic_feasibility}
For $k\geq1$ and any $\alpha>0$,
\eqref{eq:quadratic_cbf_constraint} is feasible for every
$x\in\mathcal C_q$ if and only if there exists $\eta\geq0$ such that
\begin{equation}
\begin{bmatrix}
D_N-\eta R_N & c_N\\
c_N^\top & \eta\rho
\end{bmatrix}\succeq0.
\label{eq:quadratic_feas_lmi}
\end{equation}
\end{proposition}

\begin{proof}
At states satisfying \eqref{eq:quadratic_zero_authority}, write
$x=x_c+N\lambda$. Feasibility is then equivalent to
\[
\lambda^\top R_N\lambda\geq\rho
\ \Longrightarrow\
\lambda^\top(D_N+\alpha R_N)\lambda
+2c_N^\top\lambda-\alpha\rho\geq0.
\]
Since the antecedent is strictly feasible, the S-lemma
\cite[App.~B.2]{Boyd} shows that this implication holds if and only if
there exists $\tau\geq0$ such that
\[
\begin{bmatrix}
D_N+(\alpha-\tau)R_N & c_N\\
c_N^\top & (\tau-\alpha)\rho
\end{bmatrix}\succeq0.
\]
Positive semidefiniteness requires $\tau\geq\alpha$. Setting
$\eta=\tau-\alpha$ gives \eqref{eq:quadratic_feas_lmi}, and the converse
follows by taking $\tau=\alpha+\eta$.
\end{proof}

When the center of the quadratic barrier is a controlled equilibrium, the LMI reduces to a simple projected condition.

\begin{corollary}
\label{cor:quadratic_equilibrium}
If $k\geq1$ and $Ax_c+Bu_c=0$ for some $u_c$, then
\eqref{eq:quadratic_cbf_constraint} is feasible throughout
$\mathcal C_q$ if and only if
\begin{equation}
N^\top(PA+A^\top P)N\succeq0.
\label{eq:quadratic_projected_condition}
\end{equation}
\end{corollary}

\begin{proof}
Since $N^\top PB=0$,
$c_N=N^\top PAx_c=-N^\top PBu_c=0$. Thus
\eqref{eq:quadratic_feas_lmi} holds iff
$D_N-\eta R_N\succeq0$ for some $\eta\geq0$, equivalently,
$D_N\succeq0$.
\end{proof}

\begin{remark}
\label{rem:quadratic_gain_independence}
Condition \eqref{eq:quadratic_feas_lmi} is independent of $\alpha$.
Thus, feasibility throughout $\mathcal C_q$ for one positive linear CBF
gain is equivalent to feasibility for every positive gain, and increasing
$\alpha$ cannot restore feasibility when the condition fails. This result
applies to a single quadratic constraint with unrestricted inputs; testing
multiple quadratic CBFs separately does not establish joint feasibility.
\end{remark}

\section{Numerical Results}
\label{sec:numerical}

This section illustrates how the constraint geometry shapes the feasibility domain and enables explicit safety filtering.

\subsection{Double-Integrator Example}

For the double integrator in Example~\ref{ex:double_integrator},
consider $h_1=x_1+x_2+1$, $h_2=x_1+1$, $h_3=-2x_2+5$,
$h_4=x_1-3x_2+6$, and $h_5=-2x_1+5$, with $u\in[-2,2]$.
Since $q=1$, Theorems~\ref{thm:directional_feasibility}
and~\ref{thm:directional_filter} give the feasibility domains
and scalar saturation filter, respectively.

The original design uses gain $1$ for relative-degree-one
constraints and gains $(1,2)$ for relative-degree-two constraints.
Its HOCBF-admissible set $\mathcal S$ contains $(2,-2.8)$,
where the constraint associated with $h_1$ requires
$u\geq2.6>2$; hence
$\mathcal S\nsubseteq\mathcal X_{\mathrm{feas}}^{\rm b}$.

We fix the lower gains of $h_2,h_5$ at $1/2$ and denote the
resulting HOCBF-admissible set by $\mathcal S_{1/2}$.
Then $\mathcal S_{1/2}\subsetneq\mathcal S$, while $\mathcal C$
remains unchanged. For a common final gain $\gamma>0$ applied
to all five constraints, the regional certificate of
Theorem~\ref{thm:regional_gain} gives
$\mathcal S_{1/2}\subseteq
\mathcal X_{\mathrm{feas}}^{\rm b}(\gamma)$
if and only if $\gamma\geq7/19$.
Choosing $\gamma=0.4$ therefore guarantees persistent feasibility
and safety on $\mathcal S_{1/2}$ by
Corollary~\ref{cor:persistent_feasibility}.

We simulate $8$ seconds using $u_d=-K(x-[1\ 0]^\top)$,
with the LQR gain for $Q=I_2$, $R=0.1$, from the same eight
interior initial states of $\mathcal S_{1/2}$.
All gains remain fixed. Fig.~\ref{fig:doubleint_main} illustrates
the original infeasible region (red), the certified inclusion
after redesign, and satisfaction of the actuator limits.
The explicit filter and numerical QP agree within $10^{-12}$
at the comparison samples.

\begin{figure}[t]
    \centering
    \includegraphics[width=\columnwidth]
    {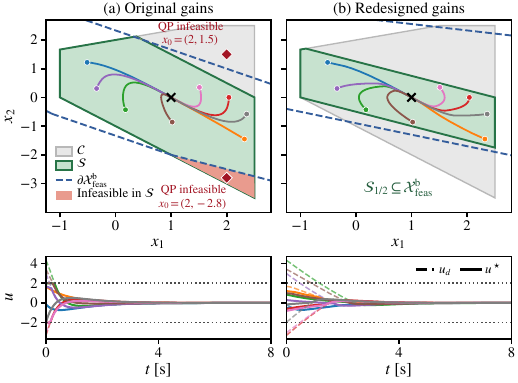}
    \caption{Original gains (left) and redesigned gains (right).
    Top: filtered trajectories (colored solid curves), initial
    states (dots), desired equilibrium (black cross), and
    feasibility boundaries (blue dashed). Red shading marks
    $\mathcal S\setminus\mathcal X_{\mathrm{feas}}^{\rm b}$;
    red diamonds mark infeasible initial states, from which
    no filtered trajectories are defined.
    Bottom: nominal (dashed) and applied (solid) inputs along
    the filtered trajectories; dotted lines mark $u=\pm2$.
    Colors match across panels.}
    \label{fig:doubleint_main}
\end{figure}

\subsection{Single Quadratic Constraint}

For the same double integrator with unrestricted input, let
$h_q(x)=x_1^2-x_1x_2+x_2^2-1$.
Here $x_c=0$ and $N=[1\ \ 1/2]^\top$ give
$N^\top(PA+A^\top P)N=3/4>0$.
Corollary~\ref{cor:quadratic_equilibrium}, a special case of
Proposition~\ref{prop:quadratic_feasibility}, therefore certifies
feasibility throughout $\mathcal C_q$ for every $\alpha>0$.

We use $\alpha=1$ and the same LQR weights with reference
$[2\ 0]^\top$, an equilibrium of both the nominal and filtered
closed loops. Fig.~\ref{fig:doubleint_quadratic}(a) illustrates
the result with five interior initial states and
$(-2/\sqrt3,-1/\sqrt3)$, where $h_q=0$ and $B^\top Px=0$.
All filtered trajectories remain outside the shaded ellipse
over $8$ seconds, whereas the unfiltered comparison enters it.
Panel~(b) shows the corresponding input corrections.

\begin{figure}[t]
    \centering
    \includegraphics[width=\columnwidth]
    {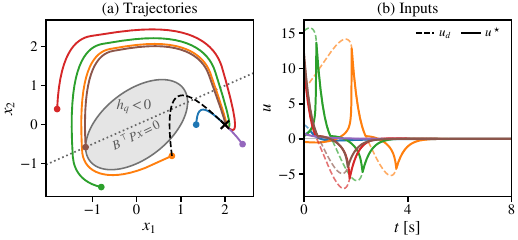}
    \caption{Quadratic constraint. (a) Filtered trajectories
    (colored solid) and an unfiltered comparison from the orange
    initial state (black dashed); the dotted line is
    $B^\top Px=0$, and the black cross marks the desired equilibrium.
    (b) Nominal inputs $u_d(x(t))$ (dashed) and applied inputs
    $u^\star(x(t))$ (solid), both evaluated along the filtered
    trajectories; colors match across panels.}
    \label{fig:doubleint_quadratic}
\end{figure}

\subsection{Aircraft Example}

We now evaluate the proposed filter on a nonlinear six-degree-of-freedom
(6-DoF) fixed-wing aircraft simulation. The three-state roll--yaw model
and baseline LQR--PI controller are those in
\cite[Sec.~14.8]{EL-KAW:24} and~\cite[Sec.~9]{lavretsky2025servo}.
The filter is designed at the \(25{,}000\)-ft, Mach-\(0.7\) trim
condition, with \(V_0=711.17~\mathrm{ft/s}\) and
\(\alpha_0=4.5627^\circ\), and is then tested on the nonlinear 6-DoF
model.

Let \(A_0,B_0\) denote the numerical roll--yaw model in
\cite{lavretsky2025servo}, with state
\([\,\beta\;p_s\;r_s\,]^\top\). For the filter design, we retain bank
angle and use \(x_p=[\,\beta\;p_s\;r_s\;\phi\,]^\top\), with
\(A_p=[\,A_0\ \ g_\phi;\ e_\phi^\top\ \ 0\,]\),
\(B_p=[\,B_0^\top\ \ 0\,]^\top\),
\(g_\phi=[0.0451323\;0\;0]^\top\), and
\(e_\phi=[0\;1.003179\;0]^\top\).
The lateral load factor is
\(N_y=C_Nx_p+D_N\delta\), where
\(C_N=[-2.604940,\;0.0187237,\;0.0676947,\;0]\),
\(D_N=[0,\;0.336976]\), and
\(\delta=[\,\delta_a\;\delta_r\,]^\top\).
The aileron and rudder satisfy
\(\ddot\delta_j+128\dot\delta_j+6400\delta_j=6400u_j\).
Thus the HOCBF design uses
\(x=[\,x_p^\top\;\delta_a\;\dot\delta_a\;
\delta_r\;\dot\delta_r\,]^\top\) and
\(u=[\,u_a\;u_r\,]^\top\).

We impose
\(|p_s|\leq2~\mathrm{deg/s}\),
\(|\beta|\leq0.25~\mathrm{deg}\), and
\(|N_y|\leq0.008~g\), with the common HOCBF gain \(\alpha=3\).
The relative degrees are \((3,3,2)\), and each output \(j\) yields
\(\ell_j^\top u+b_j^-(x)\geq0\) and
\(-\ell_j^\top u+b_j^+(x)\geq0\), with
\(b_j^\mp(x)=\pm C_j(A+3I)^{r_j}x\mp3^{r_j}y_j^{\min/\max}\).
The corresponding input-space normals are
\[
\begin{aligned}
\ell_{p_s}^\top&=[-50983.99,\;17199.86]
               =53807.09\,v_1^\top,\\
\ell_{\beta}^\top&=[0,\;97.6438]
               =97.6438\,v_2^\top,\\
\ell_{N_y}^\top&=[0,\;2156.65]
               =2156.65\,v_2^\top,
\end{aligned}
\]
where \(v_1=[-0.947533,\;0.319658]^\top\) and
\(v_2=[0,\;1]^\top\). In particular,
\(\ell_{N_y}=22.0869\,\ell_\beta\); hence the six HOCBF constraints
form only two independent directional families.

We choose \(S=[v_1\ v_2]^\top\) and
\[G=S^\top S=
\begin{bmatrix}
0.897819&-0.302886\\
-0.302886&1.102181
\end{bmatrix},\] so that
\(SG^{-1}S^\top=I_2\). Therefore, no QP is solved online:
\(z^\star=
\operatorname{sat}_{[\underline z,\overline z]}(Su_d)\) and
\(u^\star=u_d+G^{-1}S^\top(z^\star-Su_d)\).
The proposed filter thus requires only two scalar saturations and fixed
matrix operations.
All simulations use \(T=30\mathrm{s}\), \(0.01\mathrm{s}\) sampling,
zero wind, zero sideslip command, and longitudinal commands held at
trim. The roll-rate command takes the values
\([0,0,2.1,2.1,0,-2.1,-2.1,0,0]\) deg/s at
\([0,2,2.5,4,4.2,4.7,6.2,6.7,30]\) s.
Figure~\ref{fig:aircraft_compare} compares the unfiltered baseline, a
two-output closed-form filter enforcing \(p_s\) and \(\beta\) with two
independent constraint directions, and the proposed filter enforcing
\(p_s,\beta,N_y\). The square filter keeps
\(p_s\) and \(\beta\) safe but reaches \(0.00975~g\) in \(N_y\).
The proposed closed-form filter keeps all three within their limits in
this simulation, with peaks \(1.984~\mathrm{deg/s}\),
\(0.248~\mathrm{deg}\), and \(0.00771~g\).
These nonlinear sampled-data results are empirical validation; the
preceding certificates concern the continuous-time LTI design model.

\begin{figure}[t]
    \centering
    \includegraphics[width=\columnwidth]{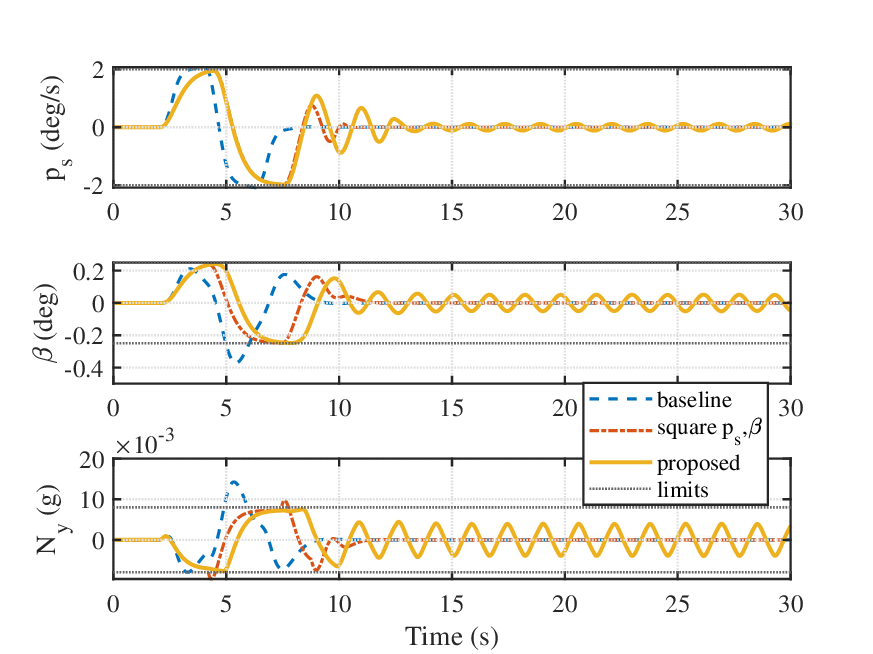}
    \caption{Nonlinear 6-DoF aircraft comparison. The square two-output
    filter enforces the \(p_s\) and \(\beta\) limits but not \(N_y\),
    whereas the proposed closed-form repeated-direction filter enforces
    all three using the same two control inputs.}
    \label{fig:aircraft_compare}
\end{figure}

\section{Conclusion}\label{sec:conclusion}
Repeated input directions yield explicit feasibility inequalities and,
under compatible actuator geometry, exact LPs for regional certification
and final-gain selection. The gain-obstruction criterion distinguishes
incompatibilities removable by tuning from those requiring redesign.
Combined with the saturation-based filter, the regional certificate
guarantees persistent feasibility and safety on the HOCBF-admissible set.
A complementary exact LMI characterizes feasibility for one exterior
ellipsoidal CBF with unrestricted inputs. %Future work will address more general state-dependent directions and multiple quadratic constraints.

\bibliography{references}

@misc{Dai2024Compatibility,
  author = {Hongkai Dai and Chuanrui Jiang and Hongchao Zhang and Andrew Clark},
  title = {Verification and Synthesis of Compatible Control Lyapunov and Control Barrier Functions},
  year = {2024},
  howpublished = {arXiv:2406.18914}
}

@misc{Fikri2026RobustScalar,
  author = {Muhamad Rausyan Fikri},
  title = {Robust Safety Filtering for Input-Constrained Underactuated Linear Systems},
  year = {2026},
  howpublished = {arXiv:2608.10872}
}

@article{mestres2025explicit,
  title={Explicit Control Barrier Function-based Safety Filters and their Resource-Aware Computation},
  author={Mestres, Pol and Mousavi, Shima Sadat and Ong, Pio and Yang, Lizhi and Das, Ersin and Burdick, Joel W and Ames, Aaron D},
  journal={IFAC World Congress},
  year={2025},
  note={To appear}
}

@Article{PM-SSM-ADA:26,
  title =        {Dynamical Properties of Safety Filters for Linear Systems and Affine Control Barrier Functions},
  author =       {P. Mestres and S. Sadat Mousavi and A. D. Ames},
  journal =      {IEEE Control Systems Letters},
  year =         2026,
  volume =       10,
  pages =        {979-984}
}

@article{mousavi2025vertices,
  title={From Vertices to Convex Hulls: Certifying Set-Wise Compatibility for CBF Constraints},
  author={Mousavi, Shima Sadat and Tan, Xiao and Ames, Aaron D},
  journal={IEEE Control Syst. Lett.},
  volume={9},
  pages={3011--3016},
  year={2025},
}

@article{WX-CB:22,
  title =        {High-order control barrier functions},
  author =       {W. Xiao and C. Belta},
  journal =     {IEEE Trans. Autom. Control},
  volume =       67,
  number =       7,
  pages =        {3655-3662},
  year =         2022
}

@inproceedings{cohen2025compatibility,
  title={Compatibility of multiple control barrier functions for constrained nonlinear systems},
  author={Cohen, Max H and Lavretsky, Eugene and Ames, Aaron D},
  booktitle={IEEE Conf. Decis. Control},
  pages={771--778},
  year={2025},
}

@inproceedings{lavretsky2025servo,
  title={Servo-controllers for linear time-invariant systems with operational constraints},
  author={Lavretsky, Eugene and Menner, Marcel},
  booktitle={Amer. Control Conf.},
  pages={4909--4916},
  year={2025},
}

@Book{EL-KAW:24,
  author =       {E. Lavretsky and K. A. Wise},
  title =        {Robust and {A}daptive {C}ontrol with {A}erospace {A}pplications},
  year =         2024,
  publisher =    {Springer},
}

@article{AmesTAC17,
  author={Ames, Aaron D. and Coogan, Samuel and Egerstedt, Magnus and Notomista, Gennaro and Sreenath, Koushil and Tabuada, Paulo},
  journal={IEEE Trans. Autom. Control},
  title={Control Barrier Functions: Theory and Applications},
  year={2017},
  volume={62},
  number={8},
  pages={3861-3876}
}

@string{AUT = "Automatica"}

@string{LCSS = "IEEE Contr. Syst. Lett."}

@article{TamasLCSS23,
    author = {T. G. Molnar and A. D. Ames},
    title = {Composing Control Barrier Functions for Complex Safety Specifications},
    journal = {IEEE Control Syst. Lett.},
    year = {2023},
volume = {7},
pages = {3615--3620},
}

@article{GlotfelterLCSS17,
author = {P. Glotfelter and J. Cort\'{e}s and M. Egerstedt},
title = {Nonsmooth Barrier Functions With Applications to Multi-Robot Systems},
journal = LCSS,
volume = {1},
number = {2},
pages = {310--315},
year = {2017},
}

@article{CortezAutomatica22,
    author = {W. S. Cortez and D. V. Dimarogonas},
    title = {Safe-by-design control for {Euler}–{Lagrange} systems},
    journal = AUT,
    year = {2022},
    volume = {146},
    pages = {110620},
}

@inproceedings{GurrietICCPS18,
  title={Towards a framework for realizable safety critical control through active set invariance},
  author={Gurriet, Thomas and Singletary, Andrew and Reher, Jacob and Ciarletta, Laurent and Feron, Eric and Ames, Aaron},
  booktitle={ACM/IEEE Int. Conf. Cyber-Phys. Syst.},
  pages={98--106},
  year={2018},
}

@article{XuAutomatica18,
    author = {X. Xu},
    title = {Constrained control of input–output linearizable systems using control sharing barrier functions},
    journal = AUT,
    year = {2018},
    volume = {87},
}

@book{Boyd,
    author = {S. Boyd and L. Vandenberghe},
    title = {Convex Optimization},
    publisher = {Cambridge Univ. Press},
    year = {2004}
}

@article{CortezLCSS22,
    author = {W. S. Cortez and X. Tan and D. V. Dimarogonas},
    title = {A Robust, Multiple Control Barrier Function Framework for Input Constrained Systems},
    journal = {IEEE Control Syst. Lett.},
    year = {2022},
    number = {6},
    pages = {1742--1747},
}

@inproceedings{TanCDC22,
    author = {X. Tan and D. V. Dimarogonas},
    title = {Compatibility checking of multiple control barrier functions for input constrained systems},
    booktitle = {IEEE Conf. Decis. Control},
    year = {2022},
    pages = {939-944},
}

@inproceedings{BreedenACC23,
    author = {J. Breeden and D. Panagou},
    title = {Compositions of Multiple Control Barrier Functions Under
Input Constraints},
    booktitle = {Amer. Control Conf.},
    year = {2023},
pages = {3688--3695},
}

@article{IsalyTAC24,
    author = {A. Isaly and M. Ghanbarpour and R. G. Sanfelice and W. E. Dixon},
    title = {On the Feasibility and Continuity of Feedback Controllers Defined by Multiple Control Barrier Functions},
    journal = {IEEE Trans. Autom. Control},
    volume = {69},
    number = {11},
    year = {2024},
    pages = {7326--7339},
}

@article{KolmanovskyAutomatica17,
    author = {E. Garone and S. Di Cairano and I. Kolmanovsky},
    title = {Reference and command governors for systems with constraints: A survey on theory and applications},
    journal = AUT,
    year = {2017},
    volume = {75},
    pages = {306--328},
}

@article{GilbertTan91,
  author  = {E. G. Gilbert and K. T. Tan},
  title   = {Linear systems with state and control constraints: The theory and application of maximal output admissible sets},
  journal = {IEEE Trans. Autom. Control},
  volume  = {36},
  number  = {9},
  pages   = {1008--1020},
  year    = {1991},
  doi     = {10.1109/9.83532}
}

@article{cohen2023smooth,
  author  = {M. H. Cohen and P. Ong and G. Bahati and A. D. Ames},
  title   = {Characterizing smooth safety filters via the implicit function theorem},
  journal = {IEEE Control Syst. Lett.},
  volume  = {7},
  pages   = {3890--3895},
  year    = {2023},
  doi     = {10.1109/LCSYS.2023.3341345}
}

\bibliographystyle{IEEEtran}

\end{document}